\documentclass[american,aps,prl,reprint,nofootinbib,longbibliography]{revtex4-2}
\usepackage{amsmath, amssymb, amsfonts, amsthm, mathtools}
\usepackage{graphicx}
\usepackage{braket}
\usepackage[pdfusetitle,
 bookmarks=true,bookmarksnumbered=false,bookmarksopen=false,
 breaklinks=false,pdfborder={0 0 0},pdfborderstyle={},backref=false,colorlinks=true]
 {hyperref}
\hypersetup{
 allcolors=magenta}
\usepackage{xcolor}
\usepackage{times,txfonts}

\newtheorem{theorem}{Theorem}
\newtheorem{lemma}{Lemma}
\newtheorem{definition}{Definition}

\newtheorem{remark}{Remark}
\newtheorem{corollary}{Corollary}
\newtheorem{proposition}{Proposition}

\begin{document}

\title{Advantage of flexible catalysis for entanglement and quantum thermodynamics}

\author{Jingsong Ao}
\author{Aby Philip}
\author{Alexander Streltsov}
\affiliation{Institute of Fundamental Technological Research, Polish Academy of Sciences, Pawi\'{n}skiego 5B, 02-106 Warsaw, Poland}

\begin{abstract}
Understanding the fundamental limits of state convertibility is crucial for establishing the boundaries of quantum information processing and thermodynamic efficiency. While auxiliary systems, catalysts, can facilitate otherwise impossible transformations, standard catalysis rigidly requires the auxiliary system to return to its exact initial state. In this work, we investigate the power of flexible catalysis, where the catalyst evolves through a cycle of states, restoring its initial configuration only after a finite number of steps. Focusing on the regime of fixed, finite dimensions, we analyze the capabilities of flexible catalysis within the resource theories of entanglement and quantum thermodynamics. In the context of entanglement, we derive conditions limiting flexible catalysts, yet show that flexible catalysis can be strictly more powerful than same-dimensional standard catalysis: it enables deterministic transformations achievable by no standard catalyst of the same dimension, and it strictly increases the success probability of stochastic local operations and classical communication. A similar deterministic advantage arises in quantum thermodynamics, where flexible catalysis enables state transformations that are impossible with any standard catalyst of fixed dimension and Hamiltonian but become achievable via flexible catalysis.
\end{abstract}

\maketitle

\textit{Introduction}---
The study of state transformations lies at the heart of quantum information theory and statistical mechanics. Determining which transformations are physically allowed under restricted operations is crucial for establishing the fundamental limits of quantum protocols. Prominent examples include transitions between pure bipartite entangled states via local operations and classical communication (LOCC) ~\cite{Nielsen_1999,RevModPhys.81.865,Chitambar_2014}, and the thermodynamic conversion of energy-incoherent states via thermal operations (TO)~\cite{Janzing2000,Horodecki_2013,PhysRevX.5.021001,Brandao_2015,PRXQuantum.3.040323} and Gibbs-preserving operations (GPO)~\cite{Faist_2015}. The fundamental limits of these transitions are mathematically characterized by majorization~\cite{Nielsen_1999} and thermomajorization~\cite{Horodecki_2013}, respectively. However, both criteria strictly limit the set of achievable transformations.

To circumvent some of these strict limitations, the concept of \textit{entanglement catalysis} was introduced~\cite{Jonathan_1999,Datta_2023,RevModPhys.96.025005}. Analogous to chemical catalysis, an auxiliary system, the catalyst, facilitates a transformation while remaining strictly invariant in its final state. Catalysis significantly expands the set of achievable transitions, revealing a richer landscape of state convertibility in entanglement and other quantum resource theories~\cite{Jonathan_1999,Grabowecky_2019,PhysRevA.79.054302,Brandao_2015,Datta_2023,RevModPhys.96.025005}. However, the standard definition of catalysis imposes a rigid constraint: the catalyst must return to its \textit{exact} initial state and remain uncorrelated with the system after a single step. In realistic experimental settings, such idealized control is highly demanding. Motivated by these limitations, more general frameworks have been developed in which the catalyst is allowed to build up correlations with the system during the transformation. These generalized notions of catalysis have led to significant advances in the resource theories of quantum thermodynamics~\cite{PhysRevX.8.041051,PhysRevLett.122.210402,PhysRevLett.126.150502,PhysRevLett.128.240501,PhysRevLett.132.200201} and entanglement~\cite{PhysRevLett.127.150503,PhysRevLett.127.080502,PhysRevLett.129.120506,PhysRevLett.133.250201}. Moreover, allowing for slight changes in the catalyst has been shown to lead to phenomena like ``embezzling'' entanglement, where transformations become trivial in the limit of infinite catalyst dimension~\cite{van_Dam_2003, Ng_2015}.

This raises a fundamental question for finite-dimensional resources: What if we relax the rigidity of the catalyst without assuming infinite resources? This leads to the paradigm of \textit{flexible catalysis}~\cite{weisz2025flexiblecatalysis}, where the catalyst is allowed to evolve through a cycle of states, returning to its initial configuration only after $n$ steps. Although, flexible catalysis reverts to standard catalysis when dimension restrictions are lifted~\cite{Duan_2005,RevModPhys.96.025005}, its behavior in the regime of \textit{fixed, finite dimensions}, relevant for near-term quantum devices, remains largely unexplored. Does the ability to ``borrow" resources within a cycle offer any advantage over a standard catalyst of the same size?

In this Letter, we investigate the power of flexible catalysis in both entanglement theory and quantum thermodynamics. Focusing first on entanglement theory, we establish a series of results characterizing the limits of flexible catalysts in fixed dimensions. By analyzing the constraints on the catalyst's components, we identify specific regimes where flexible catalysis yields no advantage, effectively forcing the flexible catalysts to degenerate into a standard catalyst.

Beyond these special cases, however, flexible catalysis is in general strictly more powerful than standard catalysis of the same dimension, and we demonstrate this in three settings. First, already for deterministic transformations of entangled states (majorization), we present an explicit example whose transformation is realized by a flexible cycle but by no standard catalyst of the same dimension, the latter certified using the algorithm of Ref.~\cite{Sun_2005}. Second, in probabilistic transformations between pure entangled states~\cite{Vidal_1999}, we show that a pair of two-qubit flexible catalysts can increase the success probability of a transformation strictly beyond what is possible with any standard catalyst of the same dimension. Third, we investigate flexible catalysis in thermo-majorization, which captures state transformations in quantum thermodynamics~\cite{Horodecki_2013,PRXQuantum.3.040323}: we provide an explicit example showing that flexible catalysis enables a deterministic transformation between two energy-incoherent states that is strictly impossible with any standard catalyst of fixed Hamiltonian. By unlocking new state transformations, flexible catalysis may provide a direct path to improving work extraction~\cite{Brandao_2015}.

\textit{Flexible Catalysis in Entanglement Theory}---
We begin by investigating transformation between two bipartite pure entangled states. Recall that every bipartite pure state $\ket{\psi}$ can be written as $\sum_i \sqrt{x_i} \ket{i}_A \ket{i}_B$, where $x_i \ge 0$ are Schmidt coefficients and $\{\ket{i}_{A/B}\}$ are local orthonormal bases~\cite{Horn_Johnson_1985}. The probability vector $\vec{x}$ formed by arranging the Schmidt coefficients in non-increasing order, is called the Schmidt vector. Nielsen proved~\cite{Nielsen_1999} that a deterministic transformation from $\ket{\psi}$ to $\ket{\phi}$ via LOCC is possible if and only if their Schmidt vectors satisfy $\vec{x} \prec \vec{y}$ (majorization). Formally, $\vec{x}$ is majorized by $\vec{y}$ if the sum of their $k$ largest components satisfies $\sum_{i=1}^k x_i \le \sum_{i=1}^k y_i$ for all $1 \le k < d$, with equality holding for $k=d$. Building on this, Jonathan and Plenio introduced entanglement catalysis~\cite{Jonathan_1999}: even if $\vec{x} \not\prec \vec{y}$, the transition may still occur via an auxiliary system $\vec{c}$ such that $\vec{x} \otimes \vec{c} \prec \vec{y} \otimes \vec{c}$. This system $\vec{c}$, which facilitates the process while remaining invariant, is called a catalyst.

Relaxing the above condition, the concept of \textit{flexible catalysis}~\cite{weisz2025flexiblecatalysis} allows the auxiliary system to return to its initial state only after completing an $n$-step cycle. Formally, this concept is defined for finite-dimensional probability vectors as follows:
\begin{definition}[Flexible Catalysis]\label{def:flexible_catalysis}
Let $\vec{x}$ and $\vec{y}$ be $d$-dimensional Schmidt vectors. We say $\vec{x}$ is majorized by $\vec{y}$ via flexible catalysis if there exists a sequence of $k$-dimensional Schmidt vectors $\{\vec{c}_i\}_{i=1}^n$ satisfying
\begin{equation}
    \vec{x} \otimes \vec{c}_i \prec \vec{y} \otimes \vec{c}_{i+1}, \label{eq:flexible_catalysis}
\end{equation}
subject to the condition $\vec{c}_{n+1} = \vec{c}_1$.
\end{definition}

Without dimension constraints, the direct sum $\frac{1}{n} \bigoplus_{i=1}^n \vec{c}_i$ serves as a valid standard catalyst~\cite{Duan_2005}; we prove this and discuss its dimensional cost in the Supplementary Materials. For fixed finite dimensions, the answer is more subtle. We first identify conditions under which flexible catalysis provides no advantage. We formally state this finding in the following theorem.
\begin{theorem}
\label{thm:low_dim_degeneracy}
Suppose a sequence of $k$-dimensional flexible catalysts $\{\vec{c}_i\}_{i=1}^n$ enables the transformation $\vec{x} \to \vec{y}$. At least one state $\vec{c}_j$ in the sequence is already a valid standard catalyst (i.e., $\vec{x} \otimes \vec{c}_j \prec \vec{y} \otimes \vec{c}_j$) if either of the following conditions holds: 
\begin{enumerate}
    \item[(i)] $k=2$.
    \item[(ii)] $k=3$, $x_1 = y_1$, and $x_d = y_d > 0$.
\end{enumerate}
\end{theorem}
\begin{proof}
(i) For $k=2$, probability vectors are totally ordered under majorization. Thus, any finite cycle $\{\vec{c}_i\}_{i=1}^n$ contains a minimal element, say $\vec{c}_j$, such that $\vec{c}_j \prec \vec{c}_{j-1}$ (indices modulo $n$). Since majorization is preserved under tensor products, $\vec{x} \otimes \vec{c}_j \prec \vec{x} \otimes \vec{c}_{j-1}$. By definition, $\vec{x} \otimes \vec{c}_{j-1} \prec \vec{y} \otimes \vec{c}_j$. Transitivity yields $\vec{x} \otimes \vec{c}_j \prec \vec{y} \otimes \vec{c}_j$, making $\vec{c}_j$ a valid standard catalyst.

(ii) The condition $\vec{x} \otimes \vec{c}_i \prec \vec{y} \otimes \vec{c}_{i+1}$ implies $x_1 c_{i,1} \le y_1 c_{i+1,1}$ and $x_d c_{i,k} \ge y_d c_{i+1,k}$. Since $x_1 = y_1$ and $x_d = y_d > 0$, this reduces to $c_{i,1} \le c_{i+1,1}$ and $c_{i,k} \ge c_{i+1,k}$. Applied cyclically over $n$ steps, these enforce strict equality: $c_{i,1} = c_{i+1,1}$ and $c_{i,k} = c_{i+1,k}$ for all $i$. For $k=3$, normalization ($\sum_{m=1}^3 c_{i,m} = 1$) fixes the remaining middle component, yielding $c_{i,2} = c_{i+1,2}$. Thus, $\vec{c}_i = \vec{c}_{i+1}$ for all $i$, and the cycle degenerates into a standard catalyst.
\end{proof}

Following Theorem~\ref{thm:low_dim_degeneracy}, one might expect that flexible catalysis always yields no advantage for transformations governed by majorization. We now show that this is not case.

\begin{theorem}\label{thm:det_advantage}
There exist d-dimensional incomparable Schmidt vectors $\vec{x}, \vec{y}$ and a cyclic pair of $k$-dimensional flexible catalysts $\{\vec{c}_1, \vec{c}_2\}$ satisfying $\vec{x} \otimes \vec{c}_1 \prec \vec{y} \otimes \vec{c}_2$ and $\vec{x} \otimes \vec{c}_2 \prec \vec{y} \otimes \vec{c}_1$, while no $k$-dimensional standard catalyst $\vec{c}$ satisfies $\vec{x} \otimes \vec{c} \prec \vec{y} \otimes \vec{c}$.
\end{theorem}

We provide an explicit instance with system dimension $d=4$ and catalyst dimension $k=3$:
\begin{align}
    \vec{x} &= \tfrac{1}{1500}(645, 500, 208, 147), \\
    \vec{y} &= \tfrac{1}{1500}(730, 376, 272, 122),
\end{align}
which are incomparable and satisfy $x_1 \neq y_1$ and $x_4 \neq y_4$, so that Theorem~\ref{thm:low_dim_degeneracy} does not apply. The flexible cycle
\begin{align}
    \vec{c}_1 &= \tfrac{1}{200}(89, 71, 40), \\
    \vec{c}_2 &= \tfrac{1}{200}(99, 54, 47),
\end{align}
satisfies both majorization relations of Theorem~\ref{thm:det_advantage}, as verified directly from Nielsen's criterion, and therefore realizes the transformation $\vec{x} \to \vec{y}$ via flexible catalysis. In contrast, no standard catalyst of dimension $k=3$ exists: applying the algorithm of Sun, Duan, and Ying~\cite{Sun_2005}, one can verify whether a $k$-dimensional standard catalyst exists and certify that no $k$-dimensional $\vec{c}$ satisfies $\vec{x} \otimes \vec{c} \prec \vec{y} \otimes \vec{c}$. For details about the algorithm in~\cite{Sun_2005} and its use, see Supplementary Material. Thus, flexible catalysis provides a strict \emph{deterministic} advantage over same-dimensional standard catalysis already in entanglement theory. This is not an isolated instance: further transformations spanning $d=4,5,6$ are collected in the Supplementary Material (Table~\ref{tab:instances}). Unlike the LOCC example in~\cite{weisz2025flexiblecatalysis}, where no individual member of the flexible set is itself a catalyst, yet the transformation is still achievable by an ordinary standard catalyst in a lower dimension, as we show explicitly in Supplementary Material.

Regarding flexible catalysis, we also prove that the maximally entangled state cannot belong to any sequence of flexible catalysts, and that flexible catalysis provides no advantage for system dimension $d=3$. Detailed proofs, along with further necessary conditions for flexible catalysis under standard majorization, are provided in the Supplementary Materials.

We now investigate whether flexible catalysis provides an advantage in this probabilistic regime. Moving beyond deterministic transitions, Vidal~\cite{Vidal_1999} extended Nielsen's framework to stochastic LOCC. The maximum success probability of converting a state $\ket{\psi}$ to $\ket{\phi}$ is entirely governed by their Schmidt vectors $\vec{x}$ and $\vec{y}$:
\begin{equation}
    P(\vec{x} \to \vec{y}) = \min_{l \in \{1, \dots, d\}} \frac{E_l(\vec{x})}{E_l(\vec{y})},
\end{equation}
where $E_l(\vec{x}) \coloneqq \sum_{i=l}^d x_i$ denotes the sum of the $d-l+1$ smallest Schmidt coefficients.

For an $n$-step flexible catalysis, the overall transformation succeeds only if every step $\vec{x} \otimes \vec{c}_i \to \vec{y} \otimes \vec{c}_{i+1}$ succeeds with a non-zero probability. To establish a fair comparison with single-shot standard catalysis, we define the \textit{success probability per step}, $P_{\text{flex}}$, as the maximum geometric mean over all valid cyclic sequences of $k$-dimensional Schmidt vectors $\{\vec{c}_i\}_{i=1}^n$:
\begin{equation}
    P_{\text{flex}}(\vec{x} \to \vec{y}) = \max_{\{\vec{c}_i\}}\left[\prod_{i=1}^nP(\vec{x} \otimes \vec{c}_i \to \vec{y} \otimes \vec{c}_{i+1})\right]^\frac{1}{n},
\end{equation}
subject to the condition $\vec{c}_{n+1} = \vec{c}_1$.

In contrast, standard catalysis requires a strictly invariant catalyst ($\vec{c}_i = \vec{c}$), with success probability $P_{\text{std}}(\psi \to \phi) =\max_{\vec{c}}P(\vec{x} \otimes \vec{c} \to \vec{y} \otimes \vec{c})$. Since this is a restricted case of a flexible sequence, $P_{\text{flex}} \ge P_{\text{std}}$ follows immediately. Crucially, we show that this inequality can be strict.
\begin{figure}[htbp]
    \centering
    \includegraphics[width=0.95\linewidth]{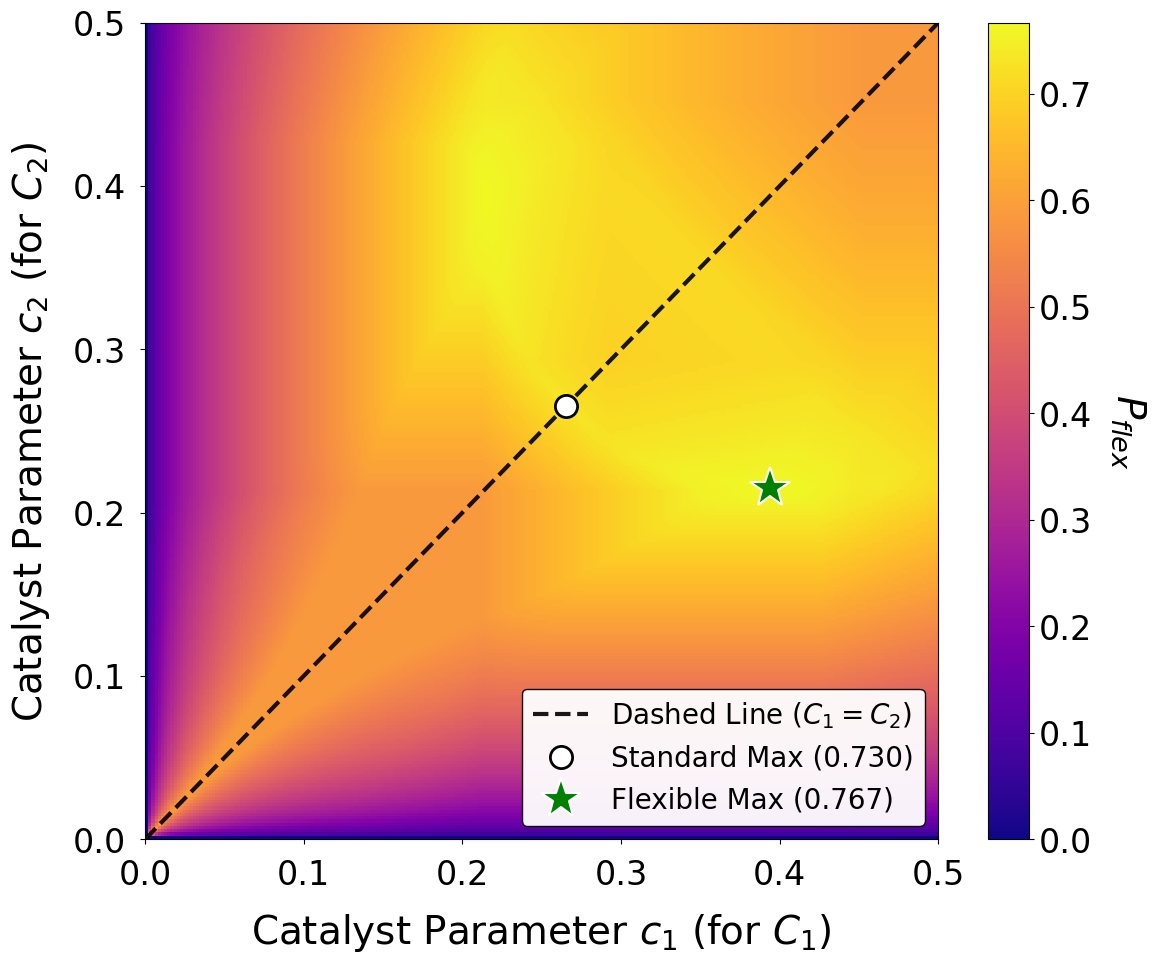}
    \caption{\textbf{Success probability landscape for a 2-step flexible catalysis of $k=2$.} The heatmap illustrates $P_{\text{flex}}$ as a function of the catalyst parameters $c_1$ and $c_2$ for the transformation $\vec{x} \to \vec{y}$. The dashed diagonal line marks the regime of standard catalysis ($c_1 = c_2$), where the maximum value is $P_{\text{std}} \approx 0.730$. The global maximum $P_{\text{flex}} \approx 0.767$ is achieved at an off-diagonal position (green star), proving the existence of a strict advantage $P_{\text{flex}} > P_{\text{std}}$.}
    \label{fig:flex_landscape}
\end{figure}
\begin{theorem}
\label{thm:strict_advantage}
There exist bipartite pure-state transformations where a flexible catalyst yields a higher success probability than any standard catalyst of the same dimension: $P_{\text{flex}} > P_{\text{std}}$.
\end{theorem}

Since a deterministic transformation is the case $P_{\text{flex}}=1$, the advantage of Theorem~\ref{thm:det_advantage} is already an instance with $P_{\text{flex}}=1>P_{\text{std}}$; Theorem~\ref{thm:strict_advantage} shows the advantage persists in the probabilistic regime $P_{\text{flex}}<1$. To demonstrate this advantage, we consider a system of dimension $d=4$ with Schmidt vectors:
\begin{align}
    \vec{x} &= (0.5789, 0.2691, 0.0872, 0.0648), \\
    \vec{y} &= (0.4937, 0.2468, 0.2043, 0.0552).
\end{align}
Without assistance, the base transition probability is $P_{\text{base}} \approx 0.5857$. We introduce an auxiliary system of dimension $k=2$, parameterizing the Schmidt vectors in flexible catalysts of $n=2$ as $\vec{C}_1 = (1-c_1, c_1)$ and $\vec{C}_2 = (1-c_2, c_2)$. 

The resulting success probability landscape is shown in Fig.~\ref{fig:flex_landscape}. The subspace of standard catalysis, defined by the diagonal $c_1 = c_2$, reaches a local maximum of $P_{\text{std}} \approx 0.7299$ at $c_1 = c_2 \approx 0.2651$ (white circle). However, the global maximum of the landscape (green star) clearly resides at an off-diagonal point $c_1 \approx 0.3936, c_2 \approx 0.2149$. This yields a success probability per step of $P_{\text{flex}} \approx 0.7666$, representing a relative gain of approximately 5\% over the optimal standard catalyst.

\textit{Flexible Catalysis in Quantum Thermodynamics}---
We now move on to state transformations in quantum thermodynamics. Just as Nielsen's theorem governs entanglement transformations~\cite{Nielsen_1999}, deterministic state transitions in quantum thermodynamics are dictated by thermo-majorization. When considering energy-incoherent states, i.e. those that commute with the system's Hamiltonian ($[\rho, H] = 0$), general quantum thermodynamic operations simplify to stochastic matrices~\cite{Horodecki_2013,PRXQuantum.3.040323,Gour_2025}. To formalize this, we must define the allowed free operations. The fundamental set of free operations is the set of \textit{thermal operations} (TO)~\cite{Janzing2000,Horodecki_2013}. A thermal operation $\mathcal{E} \in \text{TO}(A \rightarrow A)$ is a quantum channel of the form:
\begin{equation}
    \mathcal{E}(\rho^A) = \text{Tr}_B \left[ \mathcal{U} (\rho^A \otimes \gamma^B) \mathcal{U}^\dagger \right].
\end{equation}
where the system $A$, governed by Hamiltonian $H_A$, is coupled to an arbitrary heat bath $B$, governed by $H_B$, $\mathcal{U}$ is a global unitary that satisfies $[\mathcal{U}, H_A \otimes I_B + I_A \otimes H_B] = 0$ and $\gamma^B = (1/Z)\exp{(-\beta H_B)}$ is a Gibbs state with $\beta = 1/k_B T$ and $Z=\operatorname{Tr}[\exp{(-\beta H_B)}]$. 
Since the dimension of the bath $B$ is unbounded, the set of thermal operations $\text{TO}(A \rightarrow A)$ is generally neither closed nor convex~\cite{PRXQuantum.3.040323}. To ensure a robust mathematical framework, it is convenient to define \textit{Closed Thermal Operations} (CTO)~\cite{PRXQuantum.3.040323} as the topological closure of this set. A channel $\mathcal{E}$ belongs to $\text{CTO}(A \rightarrow A)$ if and only if there exists a sequence of thermal operations $\{\mathcal{E}_k\}_{k \in \mathbb{N}}$ such that:
$$\lim_{k \rightarrow \infty} \mathcal{E}_k = \mathcal{E}.$$
An even broader class of free operations is \textit{Gibbs-Preserving Operations} (GPO)~\cite{Faist_2015}, consisting of any CPTP map $\mathcal{E}$ such that $\mathcal{E}(\gamma^A)=\gamma^A$ where  $\gamma^A$ is a Gibbs state. Since thermal operations inherently preserve the Gibbs state, we have $\text{CTO}\subset\text{GPO}$~\cite{PRXQuantum.3.040323}. Remarkably, for energy-incoherent states, the operational power of these two sets is entirely equivalent. This is formalized as follows.

\begin{theorem}[Equivalence of Thermodynamic Transitions~\cite{PRXQuantum.3.040323,Gour_2025}]
\label{thm:thermo_equivalence}
Consider a quantum system with Hamiltonian $H$ at inverse temperature $\beta$. For energy-incoherent states $\rho$ and $\sigma$ ($[\rho, H] = [\sigma, H] = 0$), let the $d$-dimensional probability vectors $\vec{p}$ and $\vec{q}$ denote their respective eigenvalue spectra. Given the Gibbs vector $\vec{\gamma}$ with components $\gamma_i = e^{-\beta E_i}/Z$, where $Z=\sum_i e^{-\beta E_i}$, the following statements are equivalent for the exact transformation $\rho \to \sigma$:
\begin{enumerate}
    \item[(i)] It is achievable via a Closed Thermal Operation (CTO)---or equivalently, via standard Thermal Operations (TO) up to arbitrary precision.
    \item[(ii)] It is achievable via a Gibbs-Preserving Operation (GPO).
    \item[(iii)] $\vec{p}$ thermo-majorizes $\vec{q}$ with respect to $\vec{\gamma}$, denoted as $\vec{p} \succ_\beta \vec{q}$.
\end{enumerate}
\end{theorem}

\begin{figure*}[htbp]
    \centering
    \includegraphics[width=\textwidth, keepaspectratio]{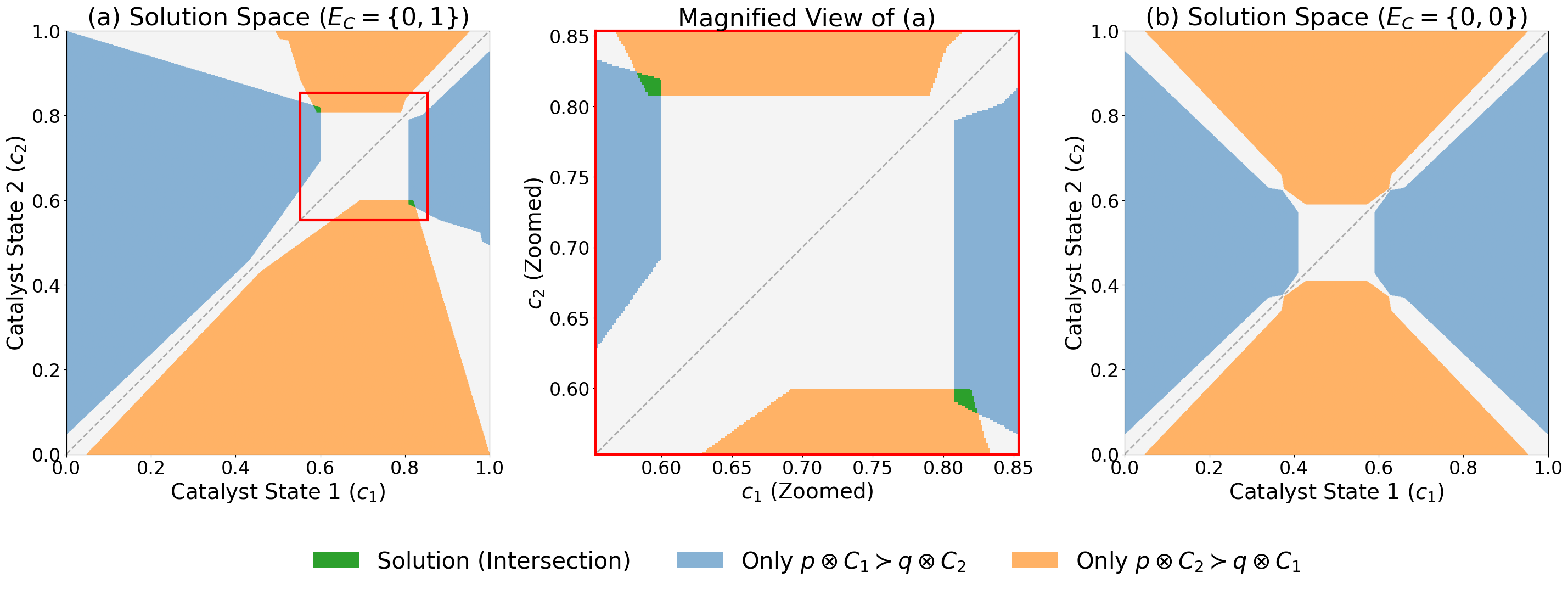}
    \caption{\textbf{Solution landscapes for flexible thermo-majorization.} Parameter space of a 2-cycle flexible catalyst pair $(c_1, c_2)$ for the transition $\vec{p} \to \vec{q}$. \textbf{(a)} For a non-trivial catalyst Hamiltonian $E_C=\{0,1\}$, a valid solution space (green region) exists off the diagonal ($c_1 \neq c_2$). The dashed diagonal line representing standard catalysis does not intersect this region. \textbf{(b)} For a degenerate Hamiltonian $E_C=\{0,0\}$, the valid region is empty.}
    \label{fig:proof_demonstration}
\end{figure*}

The \textit{thermo-majorization} relation~\cite{Horodecki_2013} is operationally defined and verified using the thermo-Lorenz curve $L_{\beta}(\vec{p})$. By applying a permutation $\pi$ that sorts the ratios $p_{\pi(i)}/\gamma_{\pi(i)}$ in non-increasing order, $L_{\beta}(\vec{p})$ is constructed as the piecewise linear curve connecting the origin to the points $(\sum_{i=1}^k \gamma_{\pi(i)}, \sum_{i=1}^k p_{\pi(i)})$ for $k=1, \dots, d$. The condition $\vec{p} \succ_\beta \vec{q}$ holds if and only if $L_{\beta}(\vec{p}) \ge L_{\beta}(\vec{q})$ everywhere. Mathematically, this condition guarantees the existence of a Gibbs-preserving stochastic matrix $M$ (satisfying $M\vec{\gamma} = \vec{\gamma}$) such that $M\vec{p} = \vec{q}$~\cite{Lostaglio_2019}.

This thermodynamic framework highlights a crucial distinction between majorization and thermo-majorization (or, more generally, relative majorization~\cite{Renes_2016}). Majorization corresponds to the special case where the reference states of both the system and the catalyst are maximally mixed. See Supplemental Materials for more details. 

Consequently, in majorization, fixing the dimension of a catalyst implicitly fixes its reference state. In the context of thermo-majorization, however, a constraint on the dimension alone is insufficient, it is also necessary to specify the explicit forms of the Hamiltonians for both the system and the catalyst.

We now turn to flexible catalysis in quantum thermodynamics, where, as in entanglement theory, it provides a deterministic advantage. Since for a two-level system no catalyst of any dimension can surpass direct thermo-majorization~\cite{CatThermalOps}, and a flexible cycle is itself reproducible by a standard catalyst in higher dimension (see Supplemental Material), a thermodynamic advantage requires $d\ge 3$, making the example below the smallest possible.

\begin{theorem}\label{thm:thermo_advantage}
For a fixed inverse temperature $\beta$ and fixed system and catalyst Hamiltonians, there exist states $\vec{p}$ and $\vec{q}$ such that no two-dimensional standard catalyst $\vec{c}$ can enable the transition ($\vec{p} \otimes \vec{c} \not\succ_{\beta} \vec{q} \otimes \vec{c}$). However, there exists a two-dimensional flexible catalyst pair $(\vec{c}_1, \vec{c}_2)$ that satisfies:
\begin{equation}\label{eqn:thermo_advantage}
    \vec{p} \otimes \vec{c}_1 \succ_{\beta} \vec{q} \otimes \vec{c}_2 \quad \text{and} \quad \vec{p} \otimes \vec{c}_2 \succ_{\beta} \vec{q} \otimes \vec{c}_1.
\end{equation}
\end{theorem}

We explicitly demonstrate this using a system with energy levels $E_S=\{0,1,2\}$, which fully characterize the system's Hamiltonian $H_S$ in the energy eigenbasis, coupled to a thermal bath at $\beta=1$, and the transition $\vec{p}=(0.09, 0.53, 0.38)^T \to \vec{q}=(0.11, 0.75, 0.14)^T$. Note that, $\vec{p}\not\succ_{\beta}\vec{q}$. We employ a two-dimensional catalyst system ($k=2$) defined similarly by its energy levels.

As shown in Fig.~\ref{fig:proof_demonstration}~(a), with an auxiliary system having non-trivial energy levels $E_C=\{0,1\}$, the transformation between $\vec{p}$ and $\vec{q}$ is possible (green region) only when $\vec{c}_1\neq \vec{c}_2$. For example, the pair $\vec{c}_1=(0.82, 0.18)^T$ and $\vec{c}_2=(0.59, 0.41)^T$ satisfies the thermo-majorization conditions in Eq.~\eqref{eqn:thermo_advantage}. The complete lack of valid transformation  between $\vec{p}$ and $\vec{q}$ using a standard catalyst, or when $\vec{c}_1 = \vec{c}_2$, confirms the failure of standard catalysis. The value of this result lies in the fact that in a practical physical situation where the catalyst system's Hamiltonian is fixed, transformations that cannot be achieved by standard catalysts may be realized using flexible catalysts. Such a constraint is physically motivated: a catalyst is realized on a physical system, like on an atomic or trapped-ion level structure~\cite{trap}, a fixed-frequency superconducting qubit~\cite{super}, or the nuclear spins of a molecule~\cite{nmr}, where energy levels are set by fabrication or by the chosen species and are not re-engineered for each target transformation; only the catalyst's state within this fixed level structure is prepared and cycled.

Crucially, while a trivial catalyst Hamiltonian ($H_C \propto \openone$) is always sufficient when the catalyst dimension is unbounded~\cite{Brandao_2015}, the non-degeneracy of the energy levels becomes necessary under finite-dimensional constraints for some transformations. We explicitly observe this in Fig.~\ref{fig:proof_demonstration}~(b): assigning degenerate catalyst energy levels ($E_C=\{0,0\}$) causes the valid solution space to vanish, showing that a non-trivial energy structure of the catalyst is necessary for this transformation.

\textit{Discussion and Conclusion}---
In this work, we have investigated the capabilities and limitations of flexible catalysis within the resource theories of entanglement and quantum thermodynamics. Our results establish that flexible catalysis provides a strict deterministic advantage over standard catalysis in both resource theories. Rather than depending on the specific physical context, we reveal that this advantage is united by a common underlying mechanism: it is fundamentally a fixed-dimension phenomenon that emerges strictly when the available dimension of the auxiliary system is restricted.

First, within the framework of entanglement theory (standard majorization), we exhibit an explicit $d=4$ pure-state transformation realized by a $k=3$ flexible cycle but admitting no standard catalyst of the same dimension; the impossibility is rigorously certified using the algorithm of Sun, Duan, and Ying~\cite{Sun_2005} (Theorem~\ref{thm:det_advantage}). This advantage persists in the probabilistic regime, where a $k=2$ flexible cycle attains a success probability strictly above the maximum reachable by any standard catalyst of the same dimension (Theorem~\ref{thm:strict_advantage}). We also delineate the strict boundaries of this effect by proving that it vanishes for catalyst dimension $k=2$ in the deterministic setting and for system dimension $d=3$ in entanglement. In the Supplementary Materials we derive further necessary conditions on the components of $\{\vec c_i\}$ in any flexible cycle, and show that a maximally entangled state can never serve as a member of flexible catalysts.

Second, in quantum thermodynamics, we demonstrate a concrete example where flexible catalysis achieves a state transformation impossible for any standard catalyst \emph{of the same dimension} with a fixed catalyst Hamiltonian, and this advantage already appears at system dimension $d=3$, the smallest possible.

Several open questions remain. First, our current analysis relies on pure bipartite states in entanglement theory and energy-incoherent states in quantum thermodynamics. The capabilities of flexible catalysis for general mixed states have yet to be thoroughly investigated. Second, in the thermodynamic context, it would be valuable to quantify the ``gap'' between flexible and standard catalysis: How does the advantage scale with the dimension of the catalyst or the complexity of the Hamiltonian spectrum? Finally, extending this analysis to the full quantum regime with coherence could reveal whether the advantages observed here persist when quantum superpositions are involved.

\textit{Acknowledgments}--- This work was supported by the National Science Centre Poland (Grant No. 2022/46/E/ST2/00115 and 2024/55/B/ST2/01590) and within the QuantERA II Programme (Grant No. 2021/03/Y/ST2/00178, acronym ExTRaQT) that has received funding from the European Union’s Horizon 2020 research and innovation programme under Grant Agreement No. 101017733.

\bibliography{references}

@Article{Jonathan_1999,
  author    = {Jonathan, Daniel and Plenio, Martin B.},
  journal   = {Physical Review Letters},
  month     = oct,
  title     = {Entanglement-Assisted Local Manipulation of Pure Quantum States},
  year      = {1999},
  issn      = {1079-7114},
  number    = {17},
  pages     = {3566--3569},
  volume    = {83},
  doi       = {10.1103/physrevlett.83.3566},
  publisher = {American Physical Society (APS)},
}

@Article{Grabowecky_2019,
  author    = {Grabowecky, Michael and Gour, Gilad},
  journal   = {Physical Review A},
  month     = {May},
  title     = {Bounds on entanglement catalysts},
  year      = {2019},
  pages     = {052348},
  volume    = {99},
  doi       = {10.1103/PhysRevA.99.052348},
  issue     = {5},
  numpages  = {8},
  publisher = {American Physical Society},
}

@Misc{weisz2025flexiblecatalysis,
  author        = {Máté Weisz and Sergii Strelchuk},
  title         = {Flexible Catalysis},
  year          = {2025},
  archiveprefix = {arXiv},
  eprint        = {2510.01065},
  eprinttype    = {arxiv},
  primaryclass  = {quant-ph},
  url           = {https://arxiv.org/abs/2510.01065}, 
}

@Article{Vidal_1999,
  author    = {Vidal, Guifré},
  journal   = {Physical Review Letters},
  month     = aug,
  title     = {Entanglement of Pure States for a Single Copy},
  year      = {1999},
  issn      = {1079-7114},
  number    = {5},
  pages     = {1046--1049},
  volume    = {83},
  doi       = {10.1103/physrevlett.83.1046},
  publisher = {American Physical Society (APS)},
}

@Article{Horodecki_2013,
  author    = {Horodecki, Michał and Oppenheim, Jonathan},
  journal   = {Nature Communications},
  month     = jun,
  title     = {Fundamental limitations for quantum and nanoscale thermodynamics},
  year      = {2013},
  issn      = {2041-1723},
  number    = {1},
  pages     = {2059},
  volume    = {4},
  doi       = {10.1038/ncomms3059},
  publisher = {Springer Science and Business Media LLC},
}

@Article{Nielsen_1999,
  author    = {Nielsen, M. A.},
  journal   = {Physical Review Letters},
  month     = jul,
  title     = {Conditions for a Class of Entanglement Transformations},
  year      = {1999},
  issn      = {1079-7114},
  number    = {2},
  pages     = {436--439},
  volume    = {83},
  doi       = {10.1103/physrevlett.83.436},
  publisher = {American Physical Society (APS)},
}

@Article{van_Dam_2003,
  author    = {van Dam, Wim and Hayden, Patrick},
  journal   = {Physical Review A},
  month     = jun,
  title     = {Universal entanglement transformations without communication},
  year      = {2003},
  issn      = {1094-1622},
  number    = {6},
  pages     = {060302},
  volume    = {67},
  doi       = {10.1103/physreva.67.060302},
  publisher = {American Physical Society (APS)},
}

@Article{Ng_2015,
  author    = {Ng, N H Y and Mančinska, L and Cirstoiu, C and Eisert, J and Wehner, S},
  journal   = {New Journal of Physics},
  month     = aug,
  title     = {Limits to catalysis in quantum thermodynamics},
  year      = {2015},
  issn      = {1367-2630},
  number    = {8},
  pages     = {085004},
  volume    = {17},
  doi       = {10.1088/1367-2630/17/8/085004},
  publisher = {IOP Publishing},
}

@Article{Brandao_2015,
  author    = {Brandão, Fernando and Horodecki, Michał and Ng, Nelly and Oppenheim, Jonathan and Wehner, Stephanie},
  journal   = {Proceedings of the National Academy of Sciences},
  month     = feb,
  title     = {The second laws of quantum thermodynamics},
  year      = {2015},
  issn      = {1091-6490},
  number    = {11},
  pages     = {3275--3279},
  volume    = {112},
  doi       = {10.1073/pnas.1411728112},
  publisher = {Proceedings of the National Academy of Sciences},
}

@Article{PhysRevX.8.041051,
  author    = {M\"uller, Markus P.},
  journal   = {Physical Review X},
  month     = {Dec},
  title     = {Correlating Thermal Machines and the Second Law at the Nanoscale},
  year      = {2018},
  pages     = {041051},
  volume    = {8},
  doi       = {10.1103/PhysRevX.8.041051},
  issue     = {4},
  numpages  = {23},
  publisher = {American Physical Society},
}

@Article{Datta_2023,
  author    = {Datta, Chandan and Varun Kondra, Tulja and Miller, Marek and Streltsov, Alexander},
  journal   = {Reports on Progress in Physics},
  month     = oct,
  title     = {Catalysis of entanglement and other quantum resources},
  year      = {2023},
  issn      = {1361-6633},
  number    = {11},
  pages     = {116002},
  volume    = {86},
  doi       = {10.1088/1361-6633/acfbec},
  publisher = {IOP Publishing},
}

@Article{Duan_2005,
  author    = {Duan, Runyao and Feng, Yuan and Li, Xin and Ying, Mingsheng},
  journal   = {Physical Review A},
  month     = apr,
  title     = {Multiple-copy entanglement transformation and entanglement catalysis},
  year      = {2005},
  issn      = {1094-1622},
  number    = {4},
  pages     = {042319},
  volume    = {71},
  doi       = {10.1103/physreva.71.042319},
  publisher = {American Physical Society (APS)},
}

@Article{Renes_2016,
  author    = {Renes, Joseph M.},
  journal   = {Journal of Mathematical Physics},
  month     = dec,
  title     = {Relative submajorization and its use in quantum resource theories},
  year      = {2016},
  issn      = {1089-7658},
  number    = {12},
  pages     = {122202},
  volume    = {57},
  doi       = {10.1063/1.4972295},
  publisher = {AIP Publishing},
}

@Article{PRXQuantum.3.040323,
  author    = {Gour, Gilad},
  journal   = {PRX Quantum},
  month     = {Nov},
  title     = {Role of Quantum Coherence in Thermodynamics},
  year      = {2022},
  pages     = {040323},
  volume    = {3},
  doi       = {10.1103/PRXQuantum.3.040323},
  issue     = {4},
  numpages  = {23},
  publisher = {American Physical Society},
}

@Article{Faist_2015,
  author    = {Faist, Philippe and Oppenheim, Jonathan and Renner, Renato},
  journal   = {New Journal of Physics},
  month     = apr,
  title     = {Gibbs-preserving maps outperform thermal operations in the quantum regime},
  year      = {2015},
  issn      = {1367-2630},
  number    = {4},
  pages     = {043003},
  volume    = {17},
  doi       = {10.1088/1367-2630/17/4/043003},
  publisher = {IOP Publishing},
}

@Book{Horn_Johnson_1985,
  author    = {Horn, Roger A. and Johnson, Charles R.},
  publisher = {Cambridge University Press},
  title     = {Matrix Analysis},
  year      = {1985},
  place     = {Cambridge},
  doi       = {10.1017/CBO9781139020411}
}

@Article{Chitambar_2014,
  author    = {Chitambar, Eric and Leung, Debbie and Mančinska, Laura and Ozols, Maris and Winter, Andreas},
  journal   = {Communications in Mathematical Physics},
  month     = mar,
  title     = {{Everything You Always Wanted to Know About LOCC (But Were Afraid to Ask)}},
  year      = {2014},
  issn      = {1432-0916},
  number    = {1},
  pages     = {303--326},
  volume    = {328},
  doi       = {10.1007/s00220-014-1953-9},
  publisher = {Springer Science and Business Media LLC},
}

@Article{Janzing2000,
  author  = {Janzing, D. and Wocjan, P. and Zeier, R. and Geiss, R. and Beth, Th.},
  journal = {International Journal of Theoretical Physics},
  month   = {Dec},
  title   = {{Thermodynamic Cost of Reliability and Low Temperatures: Tightening Landauer's Principle and the Second Law}},
  year    = {2000},
  issn    = {1572-9575},
  number  = {12},
  pages   = {2717--2753},
  volume  = {39},
  day     = {01},
  doi     = {10.1023/A:1026422630734},
}

@Book{Gour_2025,
  author    = {Gour, Gilad},
  publisher = {Cambridge University Press},
  title     = {Quantum Resource Theories},
  year      = {2025},
  isbn      = {9781009560917},
  month     = apr,
  doi       = {10.1017/9781009560870},
}

@Article{RevModPhys.81.865,
  author    = {Horodecki, Ryszard and Horodecki, Pawe\l{} and Horodecki, Micha\l{} and Horodecki, Karol},
  journal   = {Reviews of Modern Physics},
  month     = {Jun},
  title     = {Quantum entanglement},
  year      = {2009},
  pages     = {865--942},
  volume    = {81},
  doi       = {10.1103/RevModPhys.81.865},
  issue     = {2},
  numpages  = {0},
  publisher = {American Physical Society},
}

@Article{PhysRevX.5.021001,
  author    = {Lostaglio, Matteo and Korzekwa, Kamil and Jennings, David and Rudolph, Terry},
  journal   = {Physical Review X},
  month     = {Apr},
  title     = {{Quantum Coherence, Time-Translation Symmetry, and Thermodynamics}},
  year      = {2015},
  pages     = {021001},
  volume    = {5},
  doi       = {10.1103/PhysRevX.5.021001},
  issue     = {2},
  numpages  = {11},
  publisher = {American Physical Society},
}

@Article{Lostaglio_2019,
  author    = {Lostaglio, Matteo},
  journal   = {Reports on Progress in Physics},
  month     = oct,
  title     = {An introductory review of the resource theory approach to thermodynamics},
  year      = {2019},
  issn      = {1361-6633},
  number    = {11},
  pages     = {114001},
  volume    = {82},
  doi       = {10.1088/1361-6633/ab46e5},
  publisher = {IOP Publishing},
}

@Article{RevModPhys.96.025005,
  author    = {Lipka-Bartosik, Patryk and Wilming, Henrik and Ng, Nelly H. Y.},
  journal   = {Reviews of Modern Physics},
  month     = {Jun},
  title     = {Catalysis in quantum information theory},
  year      = {2024},
  pages     = {025005},
  volume    = {96},
  doi       = {10.1103/RevModPhys.96.025005},
  issue     = {2},
  numpages  = {56},
  publisher = {American Physical Society},
}

@article{PhysRevLett.122.210402,
  title = {{Von Neumann Entropy from Unitarity}},
  author = {Boes, Paul and Eisert, Jens and Gallego, Rodrigo and M\"uller, Markus P. and Wilming, Henrik},
  journal = {Physical Review Letters},
  volume = {122},
  issue = {21},
  pages = {210402},
  numpages = {6},
  year = {2019},
  month = {May},
  publisher = {American Physical Society},
  doi = {10.1103/PhysRevLett.122.210402},
  url = {https://link.aps.org/doi/10.1103/PhysRevLett.122.210402}
}

@article{PhysRevLett.126.150502,
  title = {{Quantum Thermodynamics of Correlated-Catalytic State Conversion at Small Scale}},
  author = {Shiraishi, Naoto and Sagawa, Takahiro},
  journal = {Physical Review Letters},
  volume = {126},
  issue = {15},
  pages = {150502},
  numpages = {6},
  year = {2021},
  month = {Apr},
  publisher = {American Physical Society},
  doi = {10.1103/PhysRevLett.126.150502},
  url = {https://link.aps.org/doi/10.1103/PhysRevLett.126.150502}
}

@article{PhysRevLett.127.150503,
  title = {{Catalytic Transformations of Pure Entangled States}},
  author = {Kondra, Tulja Varun and Datta, Chandan and Streltsov, Alexander},
  journal = {Physical Review Letters},
  volume = {127},
  issue = {15},
  pages = {150503},
  numpages = {6},
  year = {2021},
  month = {Oct},
  publisher = {American Physical Society},
  doi = {10.1103/PhysRevLett.127.150503},
  url = {https://link.aps.org/doi/10.1103/PhysRevLett.127.150503}
}

@article{PhysRevLett.127.080502,
  title = {{Catalytic Quantum Teleportation}},
  author = {Lipka-Bartosik, Patryk and Skrzypczyk, Paul},
  journal = {Physical Review Letters},
  volume = {127},
  issue = {8},
  pages = {080502},
  numpages = {7},
  year = {2021},
  month = {Aug},
  publisher = {American Physical Society},
  doi = {10.1103/PhysRevLett.127.080502},
  url = {https://link.aps.org/doi/10.1103/PhysRevLett.127.080502}
}

@article{PhysRevLett.133.250201,
  title = {{Catalytic and Asymptotic Equivalence for Quantum Entanglement}},
  author = {Ganardi, Ray and Kondra, Tulja Varun and Streltsov, Alexander},
  journal = {Physical Review Letters},
  volume = {133},
  issue = {25},
  pages = {250201},
  numpages = {7},
  year = {2024},
  month = {Dec},
  publisher = {American Physical Society},
  doi = {10.1103/PhysRevLett.133.250201},
  url = {https://link.aps.org/doi/10.1103/PhysRevLett.133.250201}
}

@article{PhysRevLett.129.120506,
  title = {{Fundamental Limits on Correlated Catalytic State Transformations}},
  author = {Rubboli, Roberto and Tomamichel, Marco},
  journal = {Physical Review Letters},
  volume = {129},
  issue = {12},
  pages = {120506},
  numpages = {7},
  year = {2022},
  month = {Sep},
  publisher = {American Physical Society},
  doi = {10.1103/PhysRevLett.129.120506},
  url = {https://link.aps.org/doi/10.1103/PhysRevLett.129.120506}
}

@article{PhysRevLett.132.200201,
  title = {{Coherence Manipulation in Asymmetry and Thermodynamics}},
  author = {Kondra, Tulja Varun and Ganardi, Ray and Streltsov, Alexander},
  journal = {Physical Review Letters},
  volume = {132},
  issue = {20},
  pages = {200201},
  numpages = {6},
  year = {2024},
  month = {May},
  publisher = {American Physical Society},
  doi = {10.1103/PhysRevLett.132.200201},
  url = {https://link.aps.org/doi/10.1103/PhysRevLett.132.200201}
}

@article{PhysRevLett.128.240501,
  title = {{Correlation in Catalysts Enables Arbitrary Manipulation of Quantum Coherence}},
  author = {Takagi, Ryuji and Shiraishi, Naoto},
  journal = {Physical Review Letters},
  volume = {128},
  issue = {24},
  pages = {240501},
  numpages = {7},
  year = {2022},
  month = {Jun},
  publisher = {American Physical Society},
  doi = {10.1103/PhysRevLett.128.240501},
  url = {https://link.aps.org/doi/10.1103/PhysRevLett.128.240501}
}

@article{PhysRevA.79.054302,
  title = {Necessary conditions for entanglement catalysts},
  author = {Sanders, Yuval Rishu and Gour, Gilad},
  journal = {Physical Review A},
  volume = {79},
  issue = {5},
  pages = {054302},
  numpages = {4},
  year = {2009},
  month = {May},
  publisher = {American Physical Society},
  doi = {10.1103/PhysRevA.79.054302},
  url = {https://link.aps.org/doi/10.1103/PhysRevA.79.054302}
}

@article{Sun_2005,
  title={The existence of quantum entanglement catalysts},
  author={Sun, Xiaoming and Duan, Runyao and Ying, Mingsheng},
  journal={IEEE Transactions on Information Theory},
  volume={51}, number={1}, pages={75--80}, year={2005},
  doi={10.1109/TIT.2004.839477}
}

@article{super,
  title = {Circuit quantum electrodynamics},
  author = {Blais, Alexandre and Grimsmo, Arne L. and Girvin, S. M. and Wallraff, Andreas},
  journal = {Reviews of Modern Physics},
  volume = {93},
  issue = {2},
  pages = {025005},
  numpages = {72},
  year = {2021},
  month = {May},
  publisher = {American Physical Society},
  doi = {10.1103/RevModPhys.93.025005},
  url = {https://link.aps.org/doi/10.1103/RevModPhys.93.025005}
}

@article{trap,
author = "Foss-Feig, Michael and Pagano, Guido and Potter, Andrew C. and Yao, Norman Y.",
   title = "Progress in Trapped-Ion Quantum Simulation", 
   journal= "Annual Review of Condensed Matter Physics",
   year = "2025",
   volume = "16",
   number = "Volume 16, 2025",
   pages = "145-172",
   doi = "https://doi.org/10.1146/annurev-conmatphys-032822-045619",
   url = "https://www.annualreviews.org/content/journals/10.1146/annurev-conmatphys-032822-045619",
   publisher = "Annual Reviews",
   issn = "1947-5462"
}

@article{nmr,
  title = {{NMR techniques for quantum control and computation}},
  author = {Vandersypen, L. M. K. and Chuang, I. L.},
  journal = {Reviews of Modern Physics},
  volume = {76},
  issue = {4},
  pages = {1037--1069},
  numpages = {0},
  year = {2005},
  month = {Jan},
  publisher = {American Physical Society},
  doi = {10.1103/RevModPhys.76.1037},
  url = {https://link.aps.org/doi/10.1103/RevModPhys.76.1037}
}

@article{CatThermalOps,
  title = {Catalytic transformations for thermal operations},
  author = {Czartowski, Jakub and de Oliveira Junior, A.},
  journal = {Physical Review Research},
  volume = {6},
  issue = {3},
  pages = {033203},
  numpages = {17},
  year = {2024},
  month = {Aug},
  publisher = {American Physical Society},
  doi = {10.1103/PhysRevResearch.6.033203},
  url = {https://link.aps.org/doi/10.1103/PhysRevResearch.6.033203}
}

\newpage
\appendix
\begin{center}
    \large \textbf{Supplementary Materials}
\end{center}
\setcounter{equation}{0}
\setcounter{page}{1}

\subsection{Flexible catalysis offers no advantage without dimension restrictions}\label{app:unbounded}

We work in the unified setting of relative majorization~\cite{Renes_2016}, a generalization of majorization that governs transformations of probability distributions relative to an arbitrary reference and encompasses both entanglement and thermodynamics. Given a state $\vec{p}$ with reference $\vec{r}$, $\vec{p}$ relatively majorizes $\vec{q}$ (denoted $\vec{p} \succ_{\vec{r}} \vec{q}$) if there exists a stochastic matrix $M$ with $M\vec{p} = \vec{q}$ and $M\vec{r} = \vec{r}$. Thermo-majorization~\cite{Horodecki_2013} is the special case in which $\vec{r}$ is the Gibbs state $\vec{\gamma}$ fixed by the system Hamiltonian and inverse temperature $\beta$; ordinary majorization (entanglement) is the special case of a maximally mixed reference ($H \propto \openone$, or $\beta \to 0$). A reader concerned only with entanglement may take all references uniform throughout this subsection, so that $\succ$ is ordinary majorization and the result below is exactly the claim used in the main text~\cite{Duan_2005,PhysRevLett.126.150502}; we keep a general reference only so that the same proof also covers thermo-majorization, without repeating it.

Once the catalyst dimension is left unrestricted, a flexible cycle can always be replaced by a single standard catalyst.

\begin{proposition}\label{prop:unbounded}
Let $\vec{p}, \vec{q} \in \Delta_d$ be states with reference $\vec{r}$, and let $\{\vec{c}_i\}_{i=1}^n \subset \Delta_k$ be a flexible cycle of catalysts with reference $\vec{s}$, i.e.,
\begin{equation}
    \vec{p} \otimes \vec{c}_i \;\succ_{\vec{r} \otimes \vec{s}}\; \vec{q} \otimes \vec{c}_{i+1}, \qquad \vec{c}_{n+1} = \vec{c}_1 .
\end{equation}
Then the $nk$-dimensional effective catalyst $\vec{C} = \frac{1}{n}\bigoplus_{i=1}^n \vec{c}_i$, with reference $\vec{S} = \frac{1}{n}\bigoplus_{i=1}^n \vec{s}$, serves as a single standard catalyst for the same transformation:
\begin{equation}
    \vec{p} \otimes \vec{C} \;\succ_{\vec{r} \otimes \vec{S}}\; \vec{q} \otimes \vec{C} .
\end{equation}
\end{proposition}

\begin{proof}
Relative majorization is preserved under direct sums: if $\vec{p}_i \succ_{\vec{r}_i} \vec{q}_i$ via stochastic matrices $M_i$, then the block-diagonal matrix $M = \bigoplus_i M_i$ is also stochastic and satisfies $M\big(\bigoplus_i \vec{p}_i\big) = \bigoplus_i \vec{q}_i$ and $M\big(\bigoplus_i \vec{r}_i\big) = \bigoplus_i \vec{r}_i$. Thus, $\bigoplus_i \vec{p}_i \succ_{\bigoplus_i \vec{r}_i} \bigoplus_i \vec{q}_i$.

Writing $n\vec{C} = \bigoplus_{i=1}^n \vec{c}_i$ and utilizing the distributive property of the tensor product over the direct sum, we obtain:
\begin{equation}
    \vec{p} \otimes (n\vec{C}) = \bigoplus_{i=1}^n (\vec{p} \otimes \vec{c}_i), \qquad
    \vec{q} \otimes (n\vec{C}) = \bigoplus_{i=1}^n (\vec{q} \otimes \vec{c}_{i+1}).
\end{equation}
Notice that the second sum is simply cyclically reordered. The common global reference for these direct sums is $\bigoplus_{i=1}^n (\vec{r} \otimes \vec{s}) = \vec{r} \otimes (n\vec{S})$. Since all $n$ reference blocks are identical to $\vec{r} \otimes \vec{s}$, the cyclic reordering is equivalent to applying a permutation matrix $\Pi$ that leaves the global reference invariant. Composing this permutation with the direct-sum stochastic matrix yields a valid transformation matrix witnessing $\vec{p} \otimes (n\vec{C}) \succ_{\vec{r} \otimes (n\vec{S})} \vec{q} \otimes (n\vec{C})$. Factoring out $n$ completes the proof, giving $\vec{p} \otimes \vec{C} \succ_{\vec{r} \otimes \vec{S}} \vec{q} \otimes \vec{C}$.
\end{proof}

\begin{remark}
For uniform references, the maps $M_i$ are doubly stochastic and $\vec{C} = \frac{1}{n}\bigoplus_i \vec{c}_i$ is a valid standard catalyst, recovering the entanglement statement~\cite{Duan_2005}. For the Gibbs references $\vec{r} = \vec{\gamma}_S$ and $\vec{s} = \vec{\gamma}_C$, the same $\vec{C}$ is a valid standard catalyst under thermo-majorization along with reference $\vec{S} = \frac{1}{n}\bigoplus_i \vec{\gamma}_C$, the Gibbs state of the catalyst enlarged by a degenerate $n$-level register~\cite{PhysRevLett.126.150502}.

In either scenario, the effective catalyst $\vec{C}$ has dimension $nk$, which grows with the cycle length $n$. The reduction of a flexible cycle to a standard catalyst therefore fundamentally relies on unbounded dimensionality; at a \emph{fixed} dimension $k$ it generally breaks down. This dimensional obstruction is exactly why flexible catalysis can unlock strictly advantageous transformations in the finite-dimensional regime, establishing the deterministic advantages observed for both majorization (Theorem~\ref{thm:det_advantage}) and thermo-majorization (Theorem~\ref{thm:thermo_advantage}).
\end{remark}

\subsection{Further instances of the deterministic advantage}\label{app:instances}

To demonstrate that the deterministic advantage of Theorem~\ref{thm:det_advantage} is not an isolated occurrence, Table~\ref{tab:instances} lists ten further transformations $\vec{x} \to \vec{y}$ (for $d \in \{4, 5, 6\}$ and $k = 3$) exclusively unlocked by two-step flexible catalysis. For each instance, the flexible cycle strictly realizes the transformation ($\vec{x} \otimes \vec{c}_1 \prec \vec{y} \otimes \vec{c}_2$ and $\vec{x} \otimes \vec{c}_2 \prec \vec{y} \otimes \vec{c}_1$). Crucially, the complete absence of any standard catalyst of the same dimension is certified using the algorithm of Ref.~\cite{Sun_2005}.

\begin{table*}[t]
\centering
\caption{Ten transformations realizing the deterministic advantage of Theorem~\ref{thm:det_advantage}, with catalyst dimension $k=3$ and system dimensions $d=4,5,6$. Each transformation $\vec x\to\vec y$ is realized by the flexible cycle $(\vec c_1,\vec c_2)$, while the algorithm of Ref.~\cite{Sun_2005} certifies that no standard $k=3$ catalyst exists. In every case neither $\vec c_1$ nor $\vec c_2$ is by itself a standard catalyst. The first row is the example of the main text. All quantities are exact rationals.}
\label{tab:instances}
\renewcommand{\arraystretch}{1.35}
\begin{ruledtabular}
\color{black}
\footnotesize
\begin{tabular}{c c l l l l}
\# & $d$ & $\vec x$ & $\vec y$ & $\vec c_1$ & $\vec c_2$ \\ \hline
1 & 4 & $\tfrac{1}{1500}(645,500,208,147)$ & $\tfrac{1}{1500}(730,376,272,122)$ & $\tfrac{1}{200}(89,71,40)$ & $\tfrac{1}{200}(99,54,47)$ \\
2 & 4 & $\tfrac{1}{2000}(986,627,204,183)$ & $\tfrac{1}{2000}(1091,457,330,122)$ & $\tfrac{1}{300}(146,93,61)$ & $\tfrac{1}{300}(153,80,67)$ \\
3 & 5 & $\tfrac{1}{1500}(597,451,284,134,34)$ & $\tfrac{1}{1500}(721,313,291,147,28)$ & $\tfrac{1}{75}(39,26,10)$ & $\tfrac{1}{75}(44,19,12)$ \\
4 & 5 & $\tfrac{1}{2000}(577,462,380,305,276)$ & $\tfrac{1}{2000}(619,407,383,331,260)$ & $\tfrac{1}{150}(59,49,42)$ & $\tfrac{1}{150}(57,52,41)$ \\
5 & 5 & $\tfrac{1}{2000}(635,520,397,252,196)$ & $\tfrac{1}{2000}(675,479,357,353,136)$ & $\tfrac{1}{150}(65,43,42)$ & $\tfrac{1}{150}(63,49,38)$ \\
6 & 6 & $\tfrac{1}{3000}(885,857,608,302,202,146)$ & $\tfrac{1}{3000}(1124,811,363,353,219,130)$ & $\tfrac{1}{50}(24,15,11)$ & $\tfrac{1}{50}(21,19,10)$ \\
7 & 6 & $\tfrac{1}{2000}(624,565,430,207,105,69)$ & $\tfrac{1}{2000}(961,308,279,277,121,54)$ & $\tfrac{1}{75}(36,29,10)$ & $\tfrac{1}{75}(43,20,12)$ \\
8 & 6 & $\tfrac{1}{1500}(751,429,156,70,52,42)$ & $\tfrac{1}{1500}(881,285,143,98,87,6)$ & $\tfrac{1}{150}(100,28,22)$ & $\tfrac{1}{150}(93,41,16)$ \\
9 & 6 & $\tfrac{1}{1500}(496,435,258,210,70,31)$ & $\tfrac{1}{1500}(693,281,241,141,123,21)$ & $\tfrac{1}{200}(105,67,28)$ & $\tfrac{1}{200}(119,49,32)$ \\
10 & 6 & $\tfrac{1}{5000}(1453,1316,1069,530,326,306)$ & $\tfrac{1}{5000}(2185,811,728,604,437,235)$ & $\tfrac{1}{300}(136,105,59)$ & $\tfrac{1}{300}(142,97,61)$ \\
\end{tabular}
\end{ruledtabular}
\end{table*}

\subsection{Certifying the non-existence of a standard catalyst}\label{app:sdy}

To certify that the transformations of Theorem~\ref{thm:det_advantage} and
Table~\ref{tab:instances} admit \emph{no} standard catalyst of the stated dimension, we
use the decision procedure of Sun, Duan, and Ying~\cite{Sun_2005}, specifically their
general $k\times k$ result~\cite[Theorem~3.1]{Sun_2005}. We summarize here exactly which part of that work we invoke and how.

The problem is to decide, for incomparable $\vec{x},\vec{y}\in\Delta_d^\downarrow$ and a
fixed catalyst dimension $k$, whether there exists $\vec{c}=(c_1,\dots,c_k)\in\Delta_k$
with $\vec{x}\otimes\vec{c}\prec\vec{y}\otimes\vec{c}$. By Nielsen's criterion this requires
the partial sums of the sorted entries of $\vec{x}\otimes\vec{c}$ and $\vec{y}\otimes\vec{c}$
to satisfy $\sum_{i=1}^{l}(\vec{x}\otimes\vec{c})^{\downarrow}_i \le
\sum_{i=1}^{l}(\vec{y}\otimes\vec{c})^{\downarrow}_i$ for all $l=1,\dots,dk$. The obstacle,
as noted in~\cite{Sun_2005}, is that the sorting order of the products
$\{x_a c_b\}$ and $\{y_a c_b\}$ depends on $\vec{c}$.

The key observation of Ref.~\cite{Sun_2005} is that this order is piecewise constant: it can
change only when $\vec{c}$ crosses one of the hyperplanes
\begin{equation}
\Gamma=\bigl\{\,x_{a_1}c_{b_1}=x_{a_2}c_{b_2}\,\bigr\}\cup
       \bigl\{\,y_{a_1}c_{b_1}=y_{a_2}c_{b_2}\,\bigr\},\qquad a_1<a_2,\ b_1>b_2,
\end{equation}
of which there are only $O(d^{2})$. Intersected with the simplex $\Delta_k$, these
hyperplanes partition the parameter space into finitely many cells; within each cell the
sorted orders of $\vec{x}\otimes\vec{c}$ and $\vec{y}\otimes\vec{c}$ are fixed, so each
partial sum is a fixed linear functional of $\vec{c}$ and the $dk$ majorization inequalities,
together with the linear inequalities defining the cell and $\sum_b c_b=1$, form a linear
feasibility problem. A standard catalyst exists if and only if at least one cell is feasible;
equivalently, no standard catalyst exists if and only if every cell is infeasible. We use this algorithm to determine the existence of a catalyst of a fixed dimension to enable the transformation between $\vec{x}$ and $\vec{y}$. Note that this procedure described in Theorem~3.1 of~\cite{Sun_2005} is constructive
and returns all catalysts when they exist. 

We apply this procedure at $d\in{4,5,6}$ and $k=3$. Since majorization is invariant under permutations of the entries of $\vec{c}$, it suffices to test the ordered region $c_1\ge c_2\ge c_3$, which reduces the number of cells without affecting the verdict.

\subsection{A standard catalyst for the flexible example of Ref.~\cite{weisz2025flexiblecatalysis}}\label{app:weisz}

Theorem~5.4 of Ref.~\cite{weisz2025flexiblecatalysis} realizes the LOCC transformation $\vec{a}\to\vec{b}$, with $\vec{a}=(0.4,0.4,0.1,0.1)$ and $\vec{b}=(0.5,0.29,0.21,0)$, via the flexible set $\{\vec{a},\vec{b}\}$: although $\vec{a}\not\prec\vec{b}$, one has $\vec{a}\otimes\vec{a}\prec\vec{b}\otimes\vec{b}$ and $\vec{a}\otimes\vec{b}\prec\vec{b}\otimes\vec{a}$, while neither $\vec{a}$ nor $\vec{b}$ is a standard catalyst on its own. We note that this same transformation is, however, catalyzed by an ordinary standard catalyst $\vec{c}=(3/5,2/5)$ of dimension only two. 

Thus the transformation of Theorem~5.4 in Ref.~\cite{weisz2025flexiblecatalysis} is itself achievable by ordinary standard catalysis---indeed by a catalyst of dimension ($k=2$) smaller than its flexible members ($k=4$)---and therefore exhibits no advantage of flexible over standard catalysis \emph{at fixed dimension}. Our deterministic examples go beyond it: as certified in Supplementary Material~\ref{app:sdy}, they admit no standard catalyst of the same dimension, thereby proving such an advantage.

\subsection{Properties of Flexible Catalysis in Majorization}\label{app:necessary_conditions}

In this section, we derive necessary conditions for the existence of flexible catalysts in standard majorization and prove other structural properties.

\begin{lemma}\label{lem:boundary_conditions}
If $\vec{x}\in \Delta_d^\downarrow$ is majorized by $\vec{y}\in \Delta_d^\downarrow$ via flexible catalysts $\{\vec{C}_i\}_{i=1}^n \subset \Delta_k^\downarrow$, and $x_d, y_d > 0$, the components of flexible catalysts satisfy: 
\begin{align}
    C_{i,1} &\leqslant \frac{y_1}{x_1} C_{i+1,1}, \label{eq:boundary_max} \\
    C_{i,k} &\geqslant \frac{y_{d}}{x_{d}} C_{i+1,k}. \label{eq:boundary_min}
\end{align}
\end{lemma}
\begin{proof}
Since $\vec{x}\in \Delta_d^\downarrow$ is majorized by $\vec{y}\in \Delta_d^\downarrow$ via flexible catalysts $\{\vec{C}_i\}_{i=1}^n \subset \Delta_k^\downarrow$,
this means $\vec{x} \otimes \vec{C}_i \prec \vec{y} \otimes \vec{C}_{i+1}$. This imposes constraints on the largest and smallest components: $$ x_1 C_{i,1} \leqslant y_1 C_{i+1,1} \quad \text{and} \quad x_d C_{i,k} \geqslant y_d C_{i+1,k}. $$
Since $x_d, y_d > 0$, we can rearrange these terms to obtain the stated inequalities~\eqref{eq:boundary_max} and~\eqref{eq:boundary_min}.
\end{proof}

This necessary condition can give us interesting properties of flexible catalysts. In standard catalysis, the catalyst returns to its exact initial state, naturally preserving the number of non-zero components in the catalyst vector. While flexible catalysis allows the states to change, Lemma~\ref{lem:boundary_conditions} enforces a similar restriction.

\begin{theorem}\label{thm:same_number_of_non-zero_elements_for_flexible_catalysis}
Let $\vec{x},\vec{y}\in \Delta_d^\downarrow$ be such that $x_d, y_d > 0$ and are majorized via flexible catalysts $\{\vec{C}_i\}_{i=1}^n \subset \Delta_k^\downarrow$. If for some $i$, $C_{i,k}=0$ (the $k$-th component of $\vec{C}_i$ is equal to zero), then:
\begin{equation}
    C_{i,k} = 0 \quad\forall\,  i.
\end{equation}
Consequently, all flexible catalysts in the sequence must share the same number of non-zero components.
\end{theorem}
\begin{proof}
We utilize the condition derived in Lemma~\ref{lem:boundary_conditions}. According to Eq.~\eqref{eq:boundary_min}, we have:
\begin{equation}
    C_{i,k} \geqslant \frac{y_d}{x_d} C_{i+1,k}.
\end{equation}
If we assume $C_{j,k} = 0$ for some $j<n$, the inequality becomes $0 \geqslant y_d/x_d\cdot C_{j+1,k}$. Since $x_d, y_d > 0$, $y_d/x_d>0$, this implies $C_{j+1,k} = 0$, since the components of probability vectors are non-negative. Using this same argument repeatedly, $C_{i,k}=0$ for all $j\leqslant i < n$. By applying the cyclic condition $\vec{C}_{n+1} = \vec{C}_1$, we conclude that if the $k$-th entry vanishes for some $\vec{C}_{i}$, it must vanish for all of them.

If $C_{i,k} = 0$ for all $i$, we can truncate the vectors to dimension $k-1$. We now consider the truncated vectors $\{\vec{C'}_i\}_{i=1}^n \subset \Delta_{k-1}^\downarrow$. We repeat the above procedure until the smallest component of every flexible catalyst in the sequence is strictly positive. Consequently, all flexible catalysts in the sequence must share the same number of non-zero elements.
\end{proof}

Since all flexible catalysts in the sequence must share the same number of non-zero elements, we hereafter denote this number by $k$.

Combining Lemma \ref{lem:boundary_conditions} and Theorem \ref{thm:same_number_of_non-zero_elements_for_flexible_catalysis}, we can identify a necessary condition for state transformations via flexible catalysis.

\begin{theorem}\label{lem:necessary_conditions}
If $\vec{x}\in \Delta_d^\downarrow$ is majorized by $\vec{y}\in \Delta_d^\downarrow$ via flexible catalysts $\{\vec{C}_i\}_{i=1}^n \subset \Delta_k^\downarrow$, and $x_d, y_d > 0$,  the components of $\vec{x}$ and $\vec{y}$ satisfy:
\begin{equation}
    x_1 \leqslant y_1 \quad \text{and} \quad x_d \geqslant y_d.
\end{equation}
\end{theorem}
\begin{proof}
We analyze the inequalities by multiplying the inequalities derived in Lemma \ref{lem:boundary_conditions} over the full cycle $i=1, \dots, n$.

Using the inequality $C_{i,1} \leqslant \frac{y_1}{x_1} C_{i+1,1}$ in Eq. \eqref{eq:boundary_max} , the product over the cycle yields:
\begin{equation}
    \prod_{i=1}^n C_{i,1} \leqslant \prod_{i=1}^n \left( \frac{y_1}{x_1} C_{i+1,1} \right) = \left( \frac{y_1}{x_1} \right)^n \prod_{i=1}^n C_{i+1,1}.
\end{equation}
Due to the cyclic condition $\vec{C}_{n+1} = \vec{C}_1$, the product terms are identical: $\prod_{i=1}^n C_{i,1} = \prod_{i=1}^n C_{i+1,1}>0$, allowing us to divide by it on both sides:

\begin{equation}
    1 \leqslant \left( \frac{y_1}{x_1} \right)^n \implies x_1 \leqslant y_1.
\end{equation}

Similarly, using Eq. \eqref{eq:boundary_min} ($C_{i,k} \geqslant \frac{y_d}{x_d} C_{i+1,k}$), the product yields:
\begin{equation}
    \prod_{i=1}^n C_{i,k} \geqslant \left( \frac{y_d}{x_d} \right)^n \prod_{i=1}^n C_{i+1,k}.
\end{equation}
From Theorem \ref{thm:same_number_of_non-zero_elements_for_flexible_catalysis}, since all flexible catalysts in the sequence must share the same number of non-zero elements, and $k$ denotes this number, then we know $ \prod_{i=1}^n C_{i+1,k}= \prod_{i=1}^n C_{i,k} > 0$, allowing us to divide by it on both sides:
$$ 1 \geqslant \left( \frac{y_d}{x_d} \right)^n \implies x_d \geqslant y_d. $$
This completes the proof.
\end{proof}

These conditions in Theorem~\ref{lem:necessary_conditions} are identical to those required for standard catalysis (or even no catalysis)~\cite{Jonathan_1999}, and cannot separate flexible catalysis from standard catalysis.

Next, we consider a special case of $\vec{x}$ and $\vec{y}$ where the largest and smallest components of $\vec{x}$ and $\vec{y}$ are equal. In this case, we show that the components of the flexible catalysts are also restricted. 

\begin{corollary} \label{thm:boundary_rigidity}
If the largest and smallest components of $\vec{x}$ and $\vec{y}$ are equal (i.e., $x_1 = y_1$ and $x_d = y_d > 0$), then for any flexible catalyst with $k$ non-zero elements:
\begin{equation}
    C_{i,1} = C_{i+1,1} \quad \text{and} \quad C_{i,k} = C_{i+1,k} \quad \forall i.
\end{equation}
\end{corollary}

\begin{proof}
If $x_1 = y_1$ and $x_d = y_d$, the inequalities derived in Lemma \ref{lem:boundary_conditions} reduce to the conditions:
\begin{equation}
    C_{i,1} \leqslant C_{i+1,1} \quad \text{and} \quad C_{i,k} \geqslant C_{i+1,k}. 
\end{equation}
Applying the first inequality cyclically yields the chain:
\begin{equation}
    C_{1,1} \leqslant C_{2,1} \leqslant \dots \leqslant C_{n,1} \leqslant C_{n+1,1} = C_{1,1}.
\end{equation}
The only solution to $C_{1,1} \leqslant \dots \leqslant C_{1,1}$ is strict equality $C_{i,1} = C_{i+1,1}$ for all $i$. An identical cyclic argument applies to the smallest component sequence $\{C_{i,k}\}_i$, forcing $C_{i,k} = C_{i+1,k}$. This completes the proof.
\end{proof}

Further, for the specific case of $k=3$, the restriction imposed by Corollary~\ref{thm:boundary_rigidity} is strong enough that any flexible catalysis reduces to standard catalysis. This yields the proof for Theorem~\ref{thm:low_dim_degeneracy} presented in the main text:

\begin{proof}[Proof of Theorem \ref{thm:low_dim_degeneracy}]
From Corollary \ref{thm:boundary_rigidity}, the boundary conditions imply:
\begin{equation} \label{eq:fixed_boundaries}
    C_{i,1} = C_{1,1} \quad \text{and} \quad C_{i,3} = C_{1,3} \quad \forall i.
\end{equation}
For $k=3$, the normalization condition $\sum_{j=1}^3 C_{i,j} = 1$ allows us to express the middle component $C_{i,2}$ as:
\begin{align}
    C_{i, 2} &= 1 - C_{i, 1} - C_{i, 3} \\
             &= 1 - C_{1, 1} - C_{1, 3} \quad (\text{by Eq. \eqref{eq:fixed_boundaries}}) \nonumber \\
             &= C_{1, 2}. \nonumber
\end{align}
Since every component satisfies $C_{i,j} = C_{1,j}$, then:
\begin{equation}
    \vec{C}_i = (C_{1,1}, C_{1,2}, C_{1,3}) = \vec{C}_1 \quad \forall i.
\end{equation}
This completes the proof.
\end{proof}

Now, we analyze the role of the uniform distribution vector in flexible catalysis. In standard catalysis, the uniform distribution vector cannot serve as a catalyst for any non-trivial transformation~\cite{Jonathan_1999}. A similar result holds for flexible catalysis: no state in the sequence $\{\vec{C}_i\}$ can be the uniform distribution vector.

\begin{corollary}
    The uniform distribution vector $\vec{u} = (\frac{1}{k}, \dots, \frac{1}{k})$ cannot belong to any set of flexible catalysts enabling a non-trivial transformation.
\end{corollary}

\begin{proof}
Assume there exists a set where $\vec{C}_{i+1} = \vec{u}$. Since $\vec{u}$ is the minimal element in the majorization order (it is majorized by any other vector in $\Delta_k$), we have $\vec{u} \prec \vec{C}_i$.
Using the property that $\vec{a} \prec \vec{b} \implies \vec{a} \otimes \vec{z} \prec \vec{b} \otimes \vec{z}$, we have:
$$ \vec{x} \otimes \vec{u} \prec \vec{x} \otimes \vec{C}_i. $$
Combining this with the flexible catalysis condition $\vec{x} \otimes \vec{C}_i \prec \vec{y} \otimes \vec{C}_{i+1}$ (where $\vec{C}_{i+1} = \vec{u}$):
$$ \vec{x} \otimes \vec{u} \prec \vec{x} \otimes \vec{C}_i \prec \vec{y} \otimes \vec{u}. $$
By transitivity, $\vec{x} \otimes \vec{u} \prec \vec{y} \otimes \vec{u}$, which is equivalent to $\vec{x} \prec \vec{y}$. This implies the transformation was possible without catalysis, contradicting the assumption that catalysis is necessary.
\end{proof}

Finally, we apply the necessary conditions from Theorem~\ref{lem:necessary_conditions} to the case of system dimension $d=3$. We find that for any incomparable states (i.e., neither $\vec{x} \succ \vec{y}$ nor $\vec{y} \succ \vec{x}$), flexible catalysis provides no advantage.

\begin{theorem} \label{thm:Dimension3}
Let $\vec{x}, \vec{y} \in \Delta_3^{\downarrow}$ be incomparable vectors. There exists no set of flexible catalysts enabling the transformation $\vec{x} \text{ to } \vec{y}$ or $\vec{y} \text{ to }  \vec{x}$.
\end{theorem}

\begin{proof}
From Theorem~\ref{lem:necessary_conditions}, valid flexible catalysis requires the following inequalities:
\begin{align}
    \vec{x}\otimes\vec{c}_i \to \vec{y}\otimes\vec{c}_{i+1} \quad &\text{requires} \quad x_1 \leqslant y_1 \quad \text{and} \quad x_3 \geqslant y_3. \label{eq:req_forward} \\
    \vec{y}\otimes\vec{c'}_i \to \vec{x}\otimes\vec{c'}_{i+1} \quad &\text{requires} \quad y_1 \leqslant x_1 \quad \text{and} \quad y_3 \geqslant x_3. \label{eq:req_backward}
\end{align}
If $\vec{x}$ and $\vec{y}$ are incomparable in $d=3$, their components must satisfy one of two cases:

\textbf{Case 1:} $x_1 > y_1$ and $x_3 > y_3$.
\begin{itemize}
    \item If we have $x_1 > y_1$, this violates condition \eqref{eq:req_forward}.
    \item If we have $x_3 > y_3$ (or $y_3 < x_3$), this violates condition \eqref{eq:req_backward}.
\end{itemize}

\textbf{Case 2:} $x_1 < y_1$ and $x_3 < y_3$.
\begin{itemize}
    \item If we have $x_3 < y_3$, this violates condition \eqref{eq:req_forward}.
    \item If we have $x_1 < y_1$ (or $y_1 > x_1$), this violates condition \eqref{eq:req_backward}.
\end{itemize}
In all situations, at least one necessary inequality is violated. Thus, no flexible catalysis is possible in either direction.
\end{proof}
It is important to note that standard catalysis is a special case of flexible catalysis where all states in the sequence are identical. Theorem~\ref{thm:Dimension3} directly implies that standard catalysis is also incapable of enabling these transformations. This result is consistent with the findings in Ref.~\cite{Jonathan_1999}, which established that standard catalysis provides no advantage for incomparable states in dimension $d=3$.

\subsection{Bounds on flexible catalysts}\label{app:bounds}

Now, we consider the two inequalities concerning catalysts from Theorem 1 in~\cite{Grabowecky_2019}, which establishes a lower bound on the dimension of the standard catalyst. If these two inequalities also held for flexible catalysts, it would result in an identical lower bound on their dimension. While a lower bound only establishes a minimum threshold rather than determining the exact dimensional relationship, this would suggest that flexible catalysts implementing the same transformations cannot require a strictly lower dimension. However, we will demonstrate that while one of these inequalities holds for flexible catalysts, the other does not.

We now extend the necessary condition regarding the catalyst's ratio of largest to smallest components~\cite{Grabowecky_2019} to the flexible case.

\begin{definition}[Violation Indices]
For two probability vectors, $\vec{x}, \vec{y} \in \Delta_d^\downarrow$, we define the set of violation indices $\mathcal{L}$ as the indices where the partial sums of $\vec{x}$ exceed those of $\vec{y}$:
\begin{equation}
    \mathcal{L} \coloneqq \left\{ k \in \{1, \dots, d-1\} \;\middle|\; \sum_{i=1}^k x_i > \sum_{i=1}^k y_i \right\}.
\end{equation}
\begin{equation}
    m \coloneqq \min \mathcal{L} \quad \text{and} \quad n \coloneqq \max \mathcal{L}.
\end{equation}
\end{definition}

\begin{theorem}\label{thm:flexible_bounds}
For any incomparable vectors $\vec{x}, \vec{y} \in \Delta_d^\downarrow$ and flexible catalysts $\vec{C}_i \in \Delta_k^\downarrow$, if $\vec{x} \otimes \vec{C}_i \prec \vec{y} \otimes \vec{C}_{i+1}$, then:
\begin{equation}
      \frac{C_{i,1}}{C_{i,k}}  > \max_{l \in \mathcal{L}} \left( \frac{y_l}{y_{l+1}} \right),
\end{equation}
where $\mathcal{L}$ is the set of indices associated with the violation of standard majorization defined above.
\end{theorem}

\begin{proof}
We use proof by contradiction. Assume that for all $i$, the ratio satisfies the reverse inequality:
\begin{equation}
    \frac{C_{i+1,1}}{C_{i+1,k}} \leqslant \frac{y_l}{y_{l+1}} \implies y_{l+1} C_{i+1,1} \leqslant y_l C_{i+1,k},
\end{equation}
for the specific $l$ that maximizes the ratio. This condition implies a specific ordering for the components of $\vec{y} \otimes \vec{C}_{i+1}$. Specifically, the first $kl$ largest elements of $(\vec{y} \otimes \vec{C}_{i+1})^\downarrow$ are simply the products of the first $l$ elements of $\vec{y}$ with the entire vector $\vec{C}_{i+1}$. Thus, the partial sum is:
\begin{equation}
\begin{split}
    \sum_{j=1}^{kl} (\vec{y} \otimes \vec{C}_{i+1})^\downarrow_j &= \left( \sum_{m=1}^l y_m \right) \left( \sum_{n=1}^k C_{i+1, n} \right) \\
    &= \sum_{m=1}^l y_m.
\end{split}
\end{equation}
Here we used $\sum C_{i+1, n} = 1$.

Now consider the LHS $\vec{x} \otimes \vec{C}_i$. Recall that the uniform distribution $\vec{u} = (1/k, \dots, 1/k)$ is the minimal element in the majorization order, which implies $\vec{u} \prec \vec{C}_i$ for any catalyst state. Since majorization is preserved under tensor products, we have $\vec{x} \otimes \vec{u} \prec \vec{x} \otimes \vec{C}_i$. By the definition of majorization, the partial sums of $\vec{x} \otimes \vec{C}_i$ must upper-bound those of $\vec{x} \otimes \vec{u}$:
\begin{equation}
    \sum_{j=1}^{kl} (\vec{x} \otimes \vec{C}_i)^\downarrow_j \geqslant \sum_{j=1}^{kl} (\vec{x} \otimes \vec{u})^\downarrow_j = \sum_{m=1}^l x_m.
\end{equation}

The majorization condition $\vec{x} \otimes \vec{C}_i \prec \vec{y} \otimes \vec{C}_{i+1}$ requires the partial sums of the LHS to be less than or equal to the RHS. Substituting our derived bounds:
\begin{equation}
\begin{split}
    \sum_{m=1}^l x_m &\leqslant \sum_{j=1}^{kl} (\vec{x} \otimes \vec{C}_i)^\downarrow_j \\
    &\leqslant \sum_{j=1}^{kl} (\vec{y} \otimes \vec{C}_{i+1})^\downarrow_j = \sum_{m=1}^l y_m.
\end{split}
\end{equation}
This implies $\sum_{m=1}^l x_m \leqslant \sum_{m=1}^l y_m$. However, since $\vec{x}$ and $\vec{y}$ are incomparable (specifically $\vec{x} \nprec \vec{y}$), there must exist an index $l \in \mathcal{L}$ such that $\sum_{m=1}^l x_m > \sum_{m=1}^l y_m$. This yields a direct contradiction. This concludes the proof of the theorem.
\end{proof}

Theorem~\ref{thm:flexible_bounds} constitutes a partial extension of the necessary conditions derived for standard catalysis in~\cite{Grabowecky_2019}. Because standard catalysis strictly requires bounds on both the ratio of the largest and smallest components and the ratio of adjacent components, one might naturally hypothesize that both bounds directly generalize to the flexible regime. However, the corresponding proposed restriction on the ratio of adjacent components is given by
\begin{equation}\label{eqn:proposed_ineq}
 \max_{v \in \{1, \dots, k-1\}} \left( \frac{C_{i,v}}{C_{i,v+1}} \right) < \min \left( \frac{y_1}{y_m}, \frac{y_{n+1}}{y_d} \right)
\end{equation}
does not hold in general in the case of flexible catalysis. To show this, we present an instance where \eqref{eqn:proposed_ineq} does not hold. Consider the following example in $d=4$ where:
\begin{align}
    \vec{x} &= (0.5064, 0.2565, 0.1401, 0.0970), \\
    \vec{y} &= (0.5474, 0.2048, 0.1903, 0.0575).
\end{align}
The violation set is $\mathcal{L} = \{2\}$, yielding indices $m=n=2$. Then,
\begin{equation}
    \min \left( \frac{y_1}{y_2}, \frac{y_3}{y_4} \right) = \min(2.673, 3.309) = 2.673.
\end{equation}
However, we identify a valid flexible cycle $\{\vec{C}_1, \vec{C}_2\}$ that enables the transformation:
\begin{align}
    \vec{C}_1 &= (0.6333, 0.2667, 0.1000), \\
    \vec{C}_2 &= (0.6167, 0.2833, 0.1000).
\end{align}
While the first catalyst state $\vec{C}_1$ safely satisfies the proposed adjacent-component bound (its maximum adjacent ratio is $C_{1,2}/C_{1,3} = 2.667 < 2.673$), the state $\vec{C}_2$ exhibits a maximum ratio of adjacent components—specifically between its second and third components—that exceeds the limit: \begin{equation} \frac{C_{2,2}}{C_{2,3}} = \frac{0.2833}{0.1000} = 2.833. \end{equation} Since $2.833 >2.673$, we conclude \eqref{eqn:proposed_ineq} does not hold in general. In standard catalysis, a catalyst must strictly obey this adjacent-component bound to successfully facilitate the transformation. The fact that a flexible cycle can bypass this condition implies a relaxation of structural constraints. Specifically, flexible catalysts can possess steeper drops in their probability distributions (i.e., strictly larger ratios between neighboring components) than what is permissible for a standard catalyst achieving the identical state transformation.
\end{document}